\documentclass[11pt]{amsart}
\usepackage{amsmath, amssymb}

\theoremstyle{plain}
\newtheorem{Theorem}{Theorem}
\newtheorem{Definition}{Definition}
\newtheorem{Lemma}{Lemma}
\newtheorem{Proposition}{Proposition}

\theoremstyle{definition}

 \DeclareMathOperator{\so}{so}

\title[Chaplygin Kirchhoff]
{On the cases of Kirchhoff and  Chaplygin  of the Kirchhoff
equations\footnote{MCS: 70E40, 70E45, 70E20}}

\author[Dragovi\' c, Gaji\' c]
{Vladimir Dragovi\' c{$^{1,2}$}, Borislav Gaji\'c$^1$ }
\address{$^1$\,Mathematical Institute \\Serbian Academy of Science and Art \\Kneza Mihaila 36, 11000 Belgrade\\Serbia}
\address{$^2$\, GFM, University of Lisbon, Portugal}
\email{vladad@mi.sanu.ac.rs and gajab@mi.sanu.ac.rs}

\begin{document}
\abstract It is proven that the general Kirchhoff case of the
Kirchhoff equations for $B\ne 0$ is not algebraic complete
integrable system. Similar analytic behavior of the general
solution of the Chaplygin case is detected. Four-dimensional
analogues of the Kirchhoff and the Chaplygin cases are defined on
$e(4)$ with the standard Lie-Poisson bracket.
\endabstract
\maketitle

\section{Introduction}

Starting from the fact that both in the Euler and the Lagrange cases
of motion of a heavy rigid body fixed at a point, the general
solutions are meromorphic as functions of complex time, Sofia
Kowalevski in \cite{Kow} had proposed a problem of finding all cases
of such rigid-body motion that have general solutions with the same
analytic properties. Kowalevski had discovered that there was only
one extra case,  with such dynamics, what is nowadays called {\it
the Kowalevski case}. Moreover, Kowalevski had found an additional
polynomial first integral (of fourth degree) and she solved the
system in terms of genus two theta-functions. In mid 1970's Kozlov
proved (see\cite{Koz}) a nonexistence of fourth analytical first
integral in heavy rigid-body dynamics, except in the three cases
mentioned above. Thus, it appears that in the case of motion of a
heavy rigid body fixed at a point, a fourth first integral exists
exactly in cases when the general solutions are meromorphic
functions of complex time.

Further development of such analytic theory of differential
equations was associated with Painlev\'e.

A modern theory of so-called algebraic complete integrable systems
has been performed by many authors, and the most fundamental results
 and precise definitions can be found in the book \cite{AvMV}
 (see also \cite{AvM}).

It has generally been assumed among the experts as a strong believe
that the same
parallelism between complete integrability and meromorphicity of
general solutions exists also in a similar case of motion of a rigid body in
an ideal incompressible fluid which rest in infinity, described by
so-called Kirchhoff equations. The known integrable cases are those
of Kirchhoff, the Clebsch case, the Steklov-Lyapunov case, the
Sokolov case. For the full list of known integrable cases see for
example \cite{BM}.

The main result of the present paper is the proof that the general
Kirchhoff case, for $B\ne 0$ (see below for notations) is not
algebraic complete integrable. Its general solution is not a meromorphic
function of complex time. Thus, there are classical, very well-known
cases of integrable dynamics of Kirchhoff equations which are not
algebraic integrable and which do not pass the so-called
Kowalevski-Painlev\'e test. In the proof we use the method  of  small parameter,
and historically it goes back to Poincar\'e.

We also study the Chaplygin case, which is classically known
perturbation of the Kirchhoff case and known to be non-integrable (see
\cite{KO}). By the same method of the small parameter we detect similar behavior of the
general solution of the Chaplygin case as in the general Kirchhoff
case.

Finally, we study four-dimensional generalizations and we define
appropriate analogues of the Kirchhoff and of the Chaplygin case.

\section{Three-dimensional Kirchhoff equations}

The motion of a rigid body in an ideal incompressible fluid which
rest in infinity is described by the Kirchhoff equations \cite{K, La, BM}:
\begin{equation}
\begin{aligned}
 \dot {M}&={M}\times {\frac{\partial H}{\partial M}}+
{p}\times{\frac{\partial H}{\partial p}},\qquad
\dot {p}={p}\times{\frac{\partial H}{\partial M}}\\
H&=\frac 12\langle AM,M\rangle+\langle BM, p\rangle+\frac12\langle Cp, p\rangle
\end{aligned}
\label{ke}
\end{equation}
Here $M$ and $p$ are the impulsive moment and the impulsive force.
The matrices $B$ and $C$ are symmetric, $A$ is a symmetric
positive-definite matrix, and they reflect the geometry and the
mass distribution of the body. The equations \eqref{ke} are
Hamiltonian on Lie algebra $e(3)$ with the standard Poisson
structure
$$
\{M_i,M_j\}=-\epsilon_{ijk}M_k,\quad\{M_i,p_j\}=-\epsilon_{ijk}p_k
$$
The Casimir functions are $ F_2=\langle M, p\rangle, \quad
F_3=\langle p, p\rangle$. Thus, for complete integrability in the
Liouville sense, apart from the Hamiltonian $F_1=H$, one needs one
additional first integral. The Kirchhoff case (1870) is defined by
the conditions (see \cite{K}):
$$ A=diag(a_1, a_1, a_3), \
B=diag(b_1, b_1, b_3)\ C=diag(c_1, c_1, c_3).
$$
An additional first integral is $F_4=M_3$. Thus, the Kirchhoff case is in a
sense, analogous to the Lagrange case of motion of a heavy rigid
body fixed at a point.

\section{Algebraic integrability of Kirchhoff case}

When $B=0$, the Kirchhoff case can be seen as a special case of the
Clebsch integrable case defined by $
\frac{c_2-c_3}{a_1}+\frac{c_3-c_1}{a_2}+\frac{c_1-c_2}{a_3}=0.$
A Lax representation and a higher-dimensional generalization of the
Clebsch case are given by Perelomov (see \cite{P}). It is known that
the Clebsch case is algebraic complete integrable system (see for
example \cite{L}).

Let us pass to the case when $B\ne0$. The Kowalevski-Painlev\'e test
gives  necessary conditions for an $n$-dimensional system to be
algebraic completely integrable: the system has to posses Laurent
solutions depending of $n-1$ parameters. The vector field in the
Kirchhoff equations is  homogeneous quadratic in $M_i, p_i$. One can
expect solutions which admit poles of the first order:
$M_i=\frac{M_i^{(0)}}{t}+M_i^{(1)}+...$, $p_i=\frac{p_i^{(0)}}{t}+p_i^{(1)}+...
$. By plugging it into equations \eqref{ke} one gets for
$M_i^{(0)}$ and $p_i^{(0)}$:
\begin{equation}
\begin{aligned}
-M_1^{(0)}&=(b_3-b_1)M_2^{(0)}p_3^{(0)}+(c_3-c_1)p_2^{(0)}p_3^{(0)}\\
-M_2^{(0)}&=(b_1-b_3)M_1^{(0)}p_3^{(0)}+(c_1-c_3)p_1^{(0)}p_3^{(0)}\\
-M_3^{(0)}&=0\\
-p_1^{(0)}&=-a_1M_2^{(0)}p_3^{(0)}+(b_3-b_1)p_2^{(0)}p_3^{(0)}\\
-p_2^{(0)}&=a_1M_1^{(0)}p_3^{(0)}+(b_1-b_3)p_1^{(0)}p_3^{(0)}\\
-p_3^{(0)}&=a_1(p_1^{(0)}M_2^{(0)}-p_2^{(0)}M_1^{(0)})
\end{aligned}
\label{gb}
\end{equation}
From the first integrals one can conclude that $M_1^{(0)}p_1^{(0)}+M_2^{(0)}p_2^{(0)}+M_3^{(0)}p_3^{(0)}=0$,
but from \eqref{gb} one has: $M_1^{(0)}p_1^{(0)}+M_2^{(0)}p_2^{(0)}+M_3^{(0)}p_3^{(0)}=(b_3-b_1){p_3^{(0)}}^2$.
Consequently, we have that $b_3=b_1$ or $p_3^{(0)}=0$. In the first
case one has that the mixed part $\langle BM, p\rangle$ in the
Hamiltonian is a Casimir function. Thus, this part can be dropped
from the Hamiltonian and we come to the case $B=0$. If
$p_3^{(0)}=0$, equations \eqref{gb} have only a trivial solution.
Thus, they do not admit a Laurent solution with poles of the first
order.

Thus, we have the following
\begin{Lemma}
a) In the case $B=0$ the Kirchhoff case as a special case of the
Clebsch case is algebraic complete integrable system.

b) In the case $B\ne 0$ the Kirchhoff case does not pass the
Kowalevski-Painlev\'e test.
\end{Lemma}

Thus, a  natural question arises:

{\it In the general case $B\ne 0$ is the Kirchhoff case
algebraic complete integrable system?}

In order to answer  the question, we are going to apply the method
of small parameter, following some ideas of Lyapunov (for the
details and very interesting history of the subject, see \cite{G}).

Let us choose  $b_3-b_1=\epsilon$ as our small parameter. Then, the
equations of motion in the Kirchhoff case can be rewritten:
\begin{equation}
\begin{aligned}
\dot{M}_1&=(a_3-a_1)M_2M_3+\epsilon(M_2p_3+M_3p_2)+(c_3-c_1)p_2p_3\\
\dot{M}_2&=(a_1-a_3)M_2M_3-\epsilon(M_1p_3+M_3p_1)+(c_1-c_3)p_1p_3\\
\dot{M}_3&=0\\
\dot{p}_1&=a_3M_3p_2-a_1M_2p_3+\epsilon p_2p_3\\
\dot{p}_2&=a_1M_1p_3-a_3M_3p_1-\epsilon p_1p_3\\
\dot{p}_3&=a_1(p_1M_2-p_2M_1)
\end{aligned}
\label{ks}
\end{equation}
The unperturbed system, defined by $\epsilon=0$, has  a particular
solution $M_i=\frac{M_i^{0}}{t}$, $p_i=\frac{p_i^{0}}{t}$ where:
\begin{alignat}{3}
M_1^{0}&=\frac{i\alpha}{a_1},& M_2^{0}&=\frac{i\beta}{a_1},& M_3^{0}&=0,\label{np1}\\
p_1^{0}&=\frac{i\beta}{\sqrt{a_1(c_3-c_1)}},\,\,\,& p_2^{0}&=\frac{-i\alpha}{\sqrt{a_1(c_3-c_1)}},\,\,\,& p_3^{0}&=\frac{1}{\sqrt{a_1(c_3-c_1)}}
\label{np2}
\end{alignat}
and $\alpha$, $\beta$ are constants that satisfy
$\alpha^2+\beta^2=1$. The expressions \eqref{np1}, \eqref{np2} are
the first terms in the Laurent series for a solution
$M_i=\frac{M_i^{0}}{t}+\epsilon M_i^{1}+...$,
$p_i=\frac{p_i^{0}}{t}+\epsilon p_i^{1}+...$ of the equations
\eqref{ks}. For the second terms one gets the following system:
\begin{equation}
\begin{aligned}
\dot{M}^1_1&=(a_3-a_1)\frac{M_2^0 M_3^1}{t}+\frac{M_2^0p_3^0}{t^2}+(c_3-c_1)\frac{p_2^0p_3^1}{t}+(c_3-c_1)\frac{p_2^1p_3^0}{t}\\
\dot{M}^1_2&=(a_1-a_3)\frac{M_1^0 M_3^1}{t}-\frac{M_1^0p_3^0}{t^2}+(c_1-c_3)\frac{p_1^0p_3^1}{t}+(c_1-c_3)\frac{p_1^1p_3^0}{t}\\
\dot{M}^1_3&=0\\
\dot{p}^1_1&=a_3\frac{M_3^1p_2^0}{t}-a_1\frac{M_2^0p_3^1}{t}-a_1\frac{M_2^1p_3^0}{t}+\frac{p_2^0p_3^0}{t^2}\\
\dot{p}^1_2&=-a_3\frac{M_3^1p_1^0}{t}+a_1\frac{M_1^0p_3^1}{t}+a_1\frac{M_1^1p_3^0}{t}-\frac{p_1^0p_3^0}{t^2}\\
\dot{p}^1_3&=a_1(\frac{p^0_1M^1_2}{t}+\frac{p^1_1M^0_2}{t}-\frac{p^0_2M^1_1}{t}-\frac{p^1_2M^0_1}{t})
\end{aligned}
\label{prvi}
\end{equation}
First we will find the solutions of the homogeneous system
\begin{equation}
\begin{aligned}
\dot{M}^1_1&=(c_3-c_1)\frac{p_2^0p_3^1}{t}+(c_3-c_1)\frac{p_2^1p_3^0}{t},\quad&\dot{p}^1_1&=-a_1\frac{M_2^0p_3^1}{t}-a_1\frac{M_2^1p_3^0}{t}\\
\dot{M}^1_2&=(c_1-c_3)\frac{p_1^0p_3^1}{t}+(c_1-c_3)\frac{p_1^1p_3^0}{t},\quad& \dot{p}^1_2&=a_1\frac{M_1^0p_3^1}{t}+a_1\frac{M_1^1p_3^0}{t}\\
\dot{M}^1_3&=0\quad &\dot{p}^1_3&=a_1(\frac{p^0_1M^1_2}{t}+\frac{p^1_1M^0_2}{t}-\frac{p^0_2M^1_1}{t}-\frac{p^1_2M^0_1}{t})
\end{aligned}
\label{prvihom}
\end{equation}
of the form: $M_1^1=t^sM,\ M_2^1=t^s N,\ p_1=t^s\,P,\ p_2=t^sQ, p_3=t^s R$. One gets
\begin{equation}
\begin{aligned}
sM&=(c_3-c_1)(p_2^0R+p_3^0N),\quad
sN=(c_1-c_3)(p_1^0R+p_3^0P)\\
sP&=-a_1M_2^0R-a_1p_3^0N,\qquad
sQ=a_1M_1^0R+a_1p_3^0M\\
sR&=a_1p_1^0N+a_1M_2^0P-a_1p_2^0M-a_1M_1^0Q.&&
\end{aligned}
\label{jedzas}
\end{equation}
This is a homogeneous system of linear equations and it has
nontrivial solutions only when the determinant of the system is
zero:
$$
(s-1)^3(s+1)(s+2)=0
$$
By solving the system \eqref{jedzas} for $s=1$, $s=-1$, and $s=-2$
one gets the general solution of system \eqref{prvi}:
\begin{equation}
\begin{aligned}
M_1^1&=\frac{i\alpha a}{t^2}k_5-\frac{a\beta}{t}k_4+a(k_2-i\alpha k_3)t& p_1^1&=\frac{i\beta}{t^2}k_5+\frac{\alpha}{t}k_4+k_1t\\
M_2^1&=\frac{i\beta a}{t^2}k_5+\frac{a\alpha}{t}k_4+a(-k_1-i\beta k_3)t& p_2^1&=\frac{-i\alpha}{t^2}k_5+\frac{\beta}{t}k_4+k_2t\\
&&p_3^1&=\frac{k_5}{t^2}+k_3t,
\end{aligned}
\label{opster}
\end{equation}
where $a=\sqrt{\frac{(c_3-c_1)}{a_1}}$. We will find the solution of
\eqref{prvi} by using the standard method of variation of constants.
We have:
\begin{equation*}
\begin{aligned}
\frac{i\alpha a}{t^2}k'_5-\frac{a\beta}{t}k'_4+a(k'_2-i\alpha k'_3)t&=(a_3-a_1)\frac{M_2^0 M_3^1}{t}+\frac{M_2^0p_3^0}{t^2}\\
\frac{i\beta a}{t^2}k'_5+\frac{a\alpha}{t}k'_4+a(-k'_1-i\beta k'_3)t&=(a_1-a_3)\frac{M_1^0 M_3^1}{t}-\frac{M_1^0p_3^0}{t^2}\\
\frac{i\beta}{t^2}k'_5+\frac{\alpha}{t}k'_4+k'_1t&=a_3\frac{M_3^1p_2^0}{t}+\frac{p_2^0p_3^0}{t^2}\\
\frac{-i\alpha}{t^2}k'_5+\frac{\beta}{t}k'_4+k'_2t&=-a_3\frac{M_3^1p_1^0}{t}-\frac{p_1^0p_3^0}{t^2}\\
\frac{k'_5}{t^2}+k'_3t&=0
\end{aligned}
\end{equation*}
and one gets
\begin{equation*}
\begin{aligned}
k_1(t)&=\frac{i\alpha M_3^1}{2a}\frac{1}{t}+k_1,\quad
k_2(t)=\frac{i\beta M_3^1}{2a}\frac{1}{t}+k_2,\quad
k_3(t)=k_3\\
k_4(t)&=\frac{-i(2a_3-a_1)M_3^1}{\sqrt{a_1(c_3-c_1)}}t-\frac{i}{a_1(c_3-c_1)}ln(t)+k_4,\quad
k_5(t)=k_5.
\end{aligned}
\end{equation*}
Thus, $k_4(t)$ is not a uniform function of complex time and we get
proved the following theorem:
\begin{Theorem} When $B\ne0$, the Kirchhoff case of the Kirchhoff equations is not algebraic integrable system.
\end{Theorem}

\section{Three-dimensional Chaplygin case}

In 1897 Chaplygin (see \cite{Ch}) defined a case which instead of a
first integral had an invariant relation. This system has also been
considered in  \cite{KO} where nonintegrability of the system has
been proven. The Chaplygin system is defined for $A=diag(a_1, a_2, a_3)$ by:
\begin{equation}
\begin{aligned}
b_{13}&\sqrt{a_2-a_1}\mp(b_2-b_1)\sqrt{a_3-a_2}=0,\, b_{12}=0\\
b_{13}&\sqrt{a_3-a_2}\pm(b_3-b_2)\sqrt{a_2-a_1}=0,\, b_{23}=0\\
c_{13}&\sqrt{a_2-a_1}\mp(c_2-c_1)\sqrt{a_3-a_2}=0,\, c_{12}=0\\
c_{13}&\sqrt{a_3-a_2}\pm(c_3-c_2)\sqrt{a_2-a_1}=0,\, c_{23}=0
\end{aligned}
\label{cc}
\end{equation}

The invariant relation is: $F_4=M_1\sqrt{a_2-a_1}\mp M_3\sqrt{a_3-a_2}=0$.

Conditions \eqref{cc} may be seen as certain  analogous of the
Hess-Appel'rot conditions in the case of motion of a heavy rigid
body fixed at a point. A geometric interpretation of Hess-Appel'rot
conditions has been given by Zhukovski. Using it, one can see the
Hess-Appel'rot case as a perturbation of the Lagrange top. Detailed
analysis of this system as well as higher dimensional
generalizations, the Lax representation and bi-Hamiltonian
properties are given in \cite{DG, DG2}. The class of systems of
Hess-Appel'rot type is defined there also. Starting form the basic
properties of this class od systems a new set of examples are
constructed in \cite{DGJ}.

Similarly, the Chaplygin case is a perturbation of the Kirchhoff
case. Let us choose the basis where $a_1=a_2$. In this new basis,
the Chaplygin conditions are (see for example \cite{BM}): $ a_1=a_2, a_{13}\ne 0, B=diag (b_1, b_1, b_3), C=diag (c_1, c_1, c_3)$.
The Hamiltonian becomes $2H=a_1(M_1^2+M_2^2)+a_3M_3+{2a_{13}M_1M_3}+
2b_1(M_1p_1+M_2p_2)+2b_3M_3p_3+c_1(p_1^2+p_2^2)+c_3p_3^2$. Thus $2H=H_K+2a_{13}M_1M_3$
where $H_K$ is the Hamiltonian for the Kirchhoff case. In the new
coordinates the invariant relation is $M_3=0$.

Let us apply now the  method of small parameter to the Chaplygin
case when $B=0$. We will consider $a_{13}$ as a small parameter
$\epsilon$. The equations are:
\begin{equation}
\begin{aligned}
\dot{M}_1&=(a_3-a_1)M_2M_3+(c_3-c_1)p_2p_3+\epsilon M_1M_2\\
\dot{M}_2&=(a_1-a_3)M_2M_3+(c_1-c_3)p_1p_3+\epsilon(M_3^2-M_1^2)\\
\dot{M}_3&=-\epsilon M_2M_3\\
\dot{p}_1&=a_3M_3p_2-a_1M_2p_3 +\epsilon M_1p_2\\
\dot{p}_2&=a_1M_1p_3-a_3M_3p_1+\epsilon(M_3p_3-M_1p_1)\\
\dot{p}_3&=a_1(p_1M_2-p_2M_1)-\epsilon M_3p_2
\end{aligned}
\label{ks1}
\end{equation}
We assume $M_3=0$. The unperturbed system coincides with the
unperturbed system of equations \eqref{ks}. Hence, it has a
particular solution \eqref{np1}, \eqref{np2}. For the terms of order
$\epsilon$, one gets the system which homogeneous part is
\eqref{prvihom}. Applying again the method of variation of
constants, one gets the solutions \eqref{opster} where
\begin{equation*}
\begin{aligned}
k_1(t)&=k_1,\quad
k_2(t)=k_2,\quad
k_3(t)=k_3,\\
k_4(t)&=\frac{\alpha}{a_1\sqrt{a_1(c_3-c_1)}}ln(t)+k_4,\quad k_5(t)=k_5.
\end{aligned}
\end{equation*}
From the Lax representation
for the Clebsch case, given by Perelomov in \cite{P}, a Lax
representation for the Kirchhoff case when $B=0$ can be obtained. A Lax
representation for the Chaplygin case, given below is a
perturbation of that one, given by Perelomov.

\begin{Theorem} When $B=0$, on the invariant manifold given by the invariant relation,
the equations of motion of the  Chaplygin case
are equivalent to the matrix equation:
$$
\dot{L}(\lambda)=[L(\lambda), Q(\lambda)]
$$
where $ L(\lambda)=\lambda^2L_2+\lambda L_1-L_0$,
$Q(\lambda)=\lambda Q_1+Q_0$, and
$$
L_2=diag(c_1/a_1, c_1/a_1, c_3/a_1),\  Q_1=diag(a_1, a_1, a_3)
$$
$$
L_1=\left[\begin{matrix}
0&-M_3&M_2\\
M_3&0&-M_1\\
-M_2&M_1&0
\end{matrix}
\right]
\quad
L_0=pp^{T}
$$
$$
Q_0=\left[\begin{matrix}
0&-a_3M_3-a_{13}M_1&a_1M_2\\
a_3M_3+a_{13}M_1&0&-a_1M_1-a_{13}M_3\\
-a_1M_2&a_1M_1+a_{13}M_3&0
\end{matrix}
\right]
$$
\end{Theorem}

The spectral curve $det(L(\lambda)-\mu\cdot 1)=0$ is:
$$
\begin{aligned}
\Gamma:\qquad \mu^3&+\mu^2F_3-\lambda_1^2\mu^2(c_3+2c_1)+\\
&\lambda_1^2\mu[2F_1-(2c_1+c_3)F_3]+\lambda_1^4\mu c_1(c_1+2c_3)-\\
&\lambda_1^6 c_1^2
c_3-\lambda_1^4(2c_1F_1-c_1(c_1+c_3)F_3)+\lambda_1^2a_1F_2^2=0.
\end{aligned}
$$
It is singular and has an involution $\sigma:(\lambda_1,\mu)\to
(-\lambda_1,\mu)$. The curve $\Gamma_1=\Gamma/\sigma$ is a
nonsingular genus one curve.


\section{Higher-dimensional Kirchhoff equations-the Kirchhoff and Chaplygin cases on $e(4)$}

We will consider a generalization of the Kirchhoff equations on
$e(n)$. Let us consider the Hamiltonian equations with a Hamiltonian
function:
$$
2H=\sum A_{ijkl}M_{ij}M_{kl}+2\sum B_{ijk}M_{ij}p_k+\sum C_{kl}p_kp_l
$$
in the standard Lie-Poisson structure on $so(n)$ given by:
$$
\{M_{ij}, M_{kl}\}=\delta_{ik}M_{jl}+\delta_{jl}M_{ik}-\delta_{il}M_{jk}-\delta_{jk}M_{il}
$$
$$
\{M_{ij},p_k\}=\delta_{ik}p_j-\delta_{jk}p_i
$$

We will restrict ourselves  to dimension four. A four-dimensional
Kirchhoff case should have two linear first integrals: $M_{12}$ and
$M_{34}$. It is interesting that under such assumption, the "mixed"
term in the Hamiltonian  is missing.
\begin{Proposition}
If $M_{12}$ and $M_{34}$ are the first integrals, then $B_{ijk}=0$.
\end{Proposition}
\begin{proof} Proof is based on the facts that $\dot M_{12}=\{M_{12}, H\}$, $\dot M_{34}=\{M_{34}, H\}$ and
follows by a direct calculations.
\end{proof}
\begin{Definition}
The four-dimensional Kirchhoff case is defined by
$$
\begin{aligned}
2H_K=&A_{1212}M_{12}^2+A_{1313}(M_{13}^2+M_{14}^2+M_{23}^2+M_{24}^2)+A_{3434}M_{34}^2+\\
&A_{1234}M_{12}M_{34}+C_{11}(p_1^2+p_2^2)+C_{33}(p_3^2+p_4^2)
\end{aligned}
$$
\end{Definition}
\begin{Theorem} The four dimensional Kirchhoff case is completely
integrable in Liouville sense.
\end{Theorem}
\begin{proof}
On $e(4)$ the standard Lie - Poisson structure has  two Casimir
functions:
\begin{equation*}
\begin{aligned}
F_1=&p_1^2+p_2^2+p_3^2+p_4^2,\\
F_2=&(M_{13}p_4-M_{14}p_3+M_{34}p_1)^2+(M_{23}p_1+M_{12}p_{3}-M_{13}p_2)^2+\\
&(M_{24}p_1-M_{14}p_2+M_{12}p_4)^2+(M_{23}p_4+M_{34}p_2-M_{24}p_3)^2
\end{aligned}
\end{equation*}
Thus, general  symplectic leaves are 8-dimensional and for complete
integrability one needs four first integrals in involution. Except
Hamiltonian, we have two linear first integrals $F_3=M_{12}$,
$F_4=M_{34}$ and one additional quadratic first integral:
\begin{equation*}
\begin{aligned}
F_5&=a_1(M_{12}M_{34}+M_{14}M_{23}-M_{13}M_{24})^2\\
&-c_1((M_{13}p_4-M_{14}p_3+M_{34}p_1)^2+(M_{23}p_4+M_{34}p_2-M_{24}p_3)^2)\\
&-c_3((M_{23}p_1+M_{12}p_{3}-M_{13}p_2)^2+(M_{24}p_1-M_{14}p_2+M_{12}p_4)^2)
\end{aligned}
\end{equation*}
\end{proof}

In the case of four-dimensional Chaplygin case, one can naturally
assume that $M_{12}$ and $M_{34}$ are  invariant relations. From
this assumption, we get:
\begin{Definition}
The four-dimensional Chaplygin case of the Kirchhoff equations on
$e(4)$ is defined by the Hamiltonian:
$$
\begin{aligned}
2H_{Ch}=&A_{1212}M_{12}^2+A_{1313}(M_{13}^2+M_{14}^2+M_{23}^2+M_{24}^2)+A_{3434}M_{34}^2+\\
&A_{1234}M_{12}M_{34}+A_{1213}M_{12}M_{13}+A_{1214}M_{12}M_{14}+\\
&A_{1223}M_{12}M_{23}+A_{1224}M_{12}M_{24}+A_{1334}M_{13}M_{34}+\\
&A_{1434}M_{14}M_{34}+A_{2334}M_{23}M_{34}+A_{2434}M_{24}M_{34}+\\
&B_{121}M_{12}p_1+B_{122}M_{12}p_2+B_{123}M_{12}p_3+B_{124}M_{12}p_4+\\
&B_{341}M_{34}p_1+B_{342}M_{34}p_2+B_{343}M_{34}p_3+B_{344}M_{34}p_4+\\
&C_{11}(p_1^2+p_2^2)+C_{33}(p_3^2+p_4^2).
\end{aligned}
$$
\end{Definition}

One can easily check that in this case $M_{12}$ and $M_{34}$ are
really the invariant relations.

\section*{Acknowledgments}
The research was partially supported by the Serbian Ministry of Education and
Science, Project 174020 Geometry and Topology of Manifolds,
Classical Mechanics and Integrable Dynamical Systems and by the Mathematical Physics Group of the University of
Lisbon, Project Probabilistic approach to finite and infinite dimensional dynamical systems, PTDC/MAT/104173/2008. The authors
would like to express their gratitude to Academician V. V. Kozlov
for fruitful discussions and to Professor A. V. Borisov for helpful
remarks.

\end{document}